\documentclass[showkeys,aps,nofootinbib,onecolumn,floatfix]{elsarticle}

\usepackage{amsmath,amsthm,amsfonts,amssymb,times,bbm,graphicx}

\usepackage{hyperref}

\newtheorem{theorem}{Theorem}
\newtheorem{proposition}[theorem]{Proposition}
\newtheorem{lemma}[theorem]{Lemma}

\newtheorem{definition}[theorem]{Definition}

\def\RR{\mathbbm{R}}
\def\CC{\mathbbm{C}}

\def\NN{\mathbbm{N}}
\def\EE{\mathbbm{E}}

\def\Id{\mathbbm{1}}

\DeclareMathOperator{\tr}{tr}

\begin{document}

\begin{frontmatter}

\title{RIPless compressed sensing from anisotropic measurements}

\author[eth,freiburg]{R.\ Kueng\fnref{t1}}
\author[freiburg]{D.\ Gross\fnref{t1}}

\address[eth]{Institute for Theoretical Physics, ETH Z\"urich, Wolfgang-Pauli-Strasse 27, 8093 Z\"urich, Switzerland}

\address[freiburg]{Institute for Physics, University of Freiburg, Rheinstrasse 10, 79104 Freiburg, Germany}

\fntext[t1]{Contact: \href{http://www.qc.uni-freiburg.de}{www.qc.uni-freiburg.de}}


\begin{abstract} 
Compressed sensing is the art of reconstructing a sparse vector from
its inner products with respect to a small set of randomly chosen
measurement vectors. It is usually assumed that the ensemble of
measurement vectors is in \emph{isotropic position} in the sense
that the associated covariance matrix is proportional to the
identity matrix. In this paper, we establish bounds on the number of
required measurements in the \emph{anisotropic} case, where the
ensemble of measurement vectors possesses a non-trivial covariance
matrix. Essentially, we find that the required sampling rate grows
proportionally to the condition number of the covariance matrix. In
contrast to other recent contributions to this problem, our
arguments do not rely on any \emph{restricted isometry properties}
(RIP's), but rather on ideas from convex geometry which have been
systematically studied in the theory of low-rank matrix recovery.
This allows for a simple argument and slightly improved bounds, but
may lead to a worse dependency on noise (which we do not consider in
the present paper).
\end{abstract}


\end{frontmatter}

\section{Introduction and Results}

Compressed sensing is a highly active research field in statistics
and signal analysis \cite{candes_robust_2006,
candes_near-optimal_2006, donoho_compressed_2006,
wainwright_sharp_2009}. It can be thought
of as being concerned with establishing \emph{Nyquist}-type sampling
theorems for signals which are sparse, rather than band-limited. 

More precisely, let $x\in\CC^n$ be a vector with no more than $s$
non-zero entries (i.e.\ $x$ is \emph{$s$-sparse}). Suppose we have no
information about $x$ apart from its sparsity and the inner products
$\langle a_i, x \rangle, i=1,\dots , m$ between $x$ and $m\ll n$ vectors
$a_i$. The central question is: under what conditions on $m$ and the
$a_i$'s is it possible to uniquely and computationally efficiently
recover $x$? Early celebrated results \cite{candes_robust_2006, candes_near-optimal_2006, donoho_compressed_2006}
established e.g.\ that if the measurement vectors $\{a_i\}_{i=1}^m$ are
randomly chosen discrete Fourier vectors and $m=O(s \log n)$, then,
with high probability, the unknown
vector $x$ is the unique minimizer of the $\ell_1$-norm in the affine
space defined by the known inner products.

The precise statement of our results in this introductory section will
follow very closely the exhibition in
\cite{candes_probabilistic_2011}. The reason for this approach, and
the relation of the present paper with other work (in particular
\cite{rudelson_reconstruction_2011}), is stated in
Section~\ref{sec:relations}.

We make the following definitions:
Let $F$ be 
a distribution
of random vectors on $\CC^n$. Let $a_1, \dots,
a_m$ be a sequence of i.i.d.\ random vectors drawn from $F$.
Define the \emph{sampling matrix}
\begin{equation*}
	A:=\frac 1{\sqrt m} \sum_{i=1}^m e_i a_i^*,
\end{equation*}
where $e_1, \ldots, e_m$ denote the canonical basis vectors of $\mathbb{C}^m$.
Once more, let $x$ be an $s$-sparse vector. We aim to prove that with high
probability the solution $x^\star$ to the convex optimization problem
\begin{equation}\label{eqn:l1}
	\min_{\bar{x}\in\mathbb{C}^{n}}
	\left\| \bar{x}\right\|_1\quad\mbox{subject to}\quad A\bar{x}=Ax,
\end{equation}
is unique and equal to $x$ given that the number of measurements 
$m$ is large enough. 

It turns out that the required size of $m$ depends only on two simple
properties of the ensemble $F$. These are identified below:

\begin{description}
	\item[Completeness] 
	We require that the ensemble $F$ is \emph{complete} in the sense
	that the \emph{covariance matrix} $\Sigma=\EE[a a^*]^{1/2}$ is invertible. 
	The \emph{condition number}
	\footnote{
		Recall that the condition number of a matrix is the ratio
		between its largest and its smallest singular value.
	}
	of $\Sigma$ will be denoted by $\kappa$.
\end{description}

Most of the previous work has focused on the case where the covariance
matrix is proportional to the identity matrix $\Sigma\propto \Id$
(however, see Section~\ref{sec:relations}). We refer to this case as the
\emph{isotropic} one. 

In order to describe the second relevant property of the ensemble, we
have to fix a scale. Indeed, note that the minimizer of the convex
problem (\ref{eqn:l1}) is invariant under re-scaling of the ensemble
(i.e. substituting $a_i$ by $\nu a_i$ for a number $\nu \neq 0$). The
same is true for the condition number $\kappa$. Thus, we are free to
pick an advantageous scale, without affecting the notions introduced
so far. In the isotropic case, a natural normalization convention
\cite{candes_probabilistic_2011} consists in requiring that $\EE[a a^*] =
\Id$. This option is not available in the more general, anisotropic
case, we are interested in here.  Instead, we will implicitly demand
from now that
\begin{equation}\label{eqn:scaling}
	\lambda_\mathrm{max}(\EE[a a^*]) = \lambda_\mathrm{min}(\EE[a
	a^*])^{-1},
\end{equation}
where $\lambda_{\mathrm{max}}, \lambda_{\mathrm{min}}$ denote the
maximal and the minimal eigenvalue respectively. In the isotropic
case, this reduces to the normalization $\EE[a a^*]=\Id$ used in
\cite{candes_probabilistic_2011}.

The fact that  (\ref{eqn:scaling}) can always be achieved (and further
properties that follow from it) will be established in
Lemma \ref{lem:e0} below.
With this convention, we define:

\begin{description}
	\item[Incoherence]	
	The \emph{incoherence parameter} is the smallest number $\mu$ such
	that
	\begin{equation}
		\max_{1\leq i\leq n}
			\left|\left\langle a,e_{i}\right\rangle \right|^{2}
			\leq \mu, \qquad 
		\max_{1\leq i\leq n}
		\left|\left\langle a,\EE[a a^*]^{-1} \, e_{i}\right\rangle \right|^{2} 
		\leq \mu
		\label{eqn:incoherence2}
	\end{equation}
	holds almost surely.
\end{description}

The previously known isotropic result we aim to generalize is:

\begin{theorem}[\cite{candes_probabilistic_2011}]
	Let $x$ be an $s$-sparse vector in
	$\mathbb{R}^n$.
	If we demand isotropy ($\EE[aa^*]=\Id$) and if the number of
	measurements fulfills
	\begin{equation*}
		m \geq
		C_{\omega}\mu s \log n,
	\end{equation*}
	then the solution $x^\star$ of the convex program (\ref{eqn:l1}) is
	unique and equal to $x$ with probability at least 
	$1-\frac{5}{n}-\mathrm{e}^{-\omega}$.

	In the statement above, $C_{\omega}$ may be chosen as $C_0
	\left(1+\omega\right)$ for some positive numerical constant $C_0$
\end{theorem}

Our main theorem reads:

\begin{theorem}[Main Theorem]\label{thm:main}
	Let $x\in\CC^n$ be an $s$-sparse vector, let $\omega\geq1$. If the
	number of measurements fulfills
	\begin{equation*}
		m \geq C  \kappa \mu \omega^2 s \log n,
	\end{equation*}
	then the solution $x^\star$ of the convex program (\ref{eqn:l1}) is
	unique and equal to $x$ with probability at least 
	$1-\mathrm{e}^{-\omega}$.
\end{theorem}

In the statement above, $C$ is a constant less than $18044$. For $n,s$
sufficiently large, the value may be improved to $C\leq 228$. We have
made no attempts to optimize these constants. 

Comparing these two theorems, we see that the effect of dropping the
isotropy constraint on the ensemble can essentially be captured in a
single, simple quantity: the condition number $\kappa$ of the
covariance matrix. All other minor differences between Theorem~1 and
Theorem~2 result from slightly different proof techniques.

\subsection{Improvements}

A first way of improving the result is based on a definition borrowed
from
\cite[Def.\ 1.2]{rudelson_reconstruction_2011}
\footnote{
	In fact, our definition differs very slightly from
	\cite{rudelson_reconstruction_2011}:
	their $\rho_{\max}(s,X)$ is the square of our $\lambda_{\max}(s,X)$.
	We opted for this change because the notions defined here reduce to
	the ordinary eigenvalues in the case of $s=n$.
}:

\begin{definition}
	The \emph{largest and smallest $s$-sparse eigenvalue} of a matrix
	$X$ are given by
	\begin{equation*}
		\lambda_{\mathrm{max}}(s, X) :=
		\max_{v, \|v\|_0 \leq s} \frac{\|X v\|_2}{\|v\|_2}, \qquad
		\lambda_{\mathrm{min}}(s, X) :=
		\min_{v, \|v\|_0 \leq s} \frac{\|X v\|_2}{\|v\|_2},
	\end{equation*}
	where $\|v\|_0 = | \mathrm{supp}(v)|$ denotes the cardinality of the support (i.e. the sparsity) of $v$.
	The \emph{$s$-sparse condition number}
	\footnote{
		Estimating $\mathrm{cond}(s,X)$ is equivalent to computing the RIP
		constants of $X$ (c.f.\ e.g.\ \cite{juditsky_verifiable_2011}).
		There are currently no tractable methods known for computing these
		numbers for any concrete set of matrices. We want to emphasize
		that while the mathematical concept of ``RIP
		constants'' appears in our sharpened result, its use here is completely
		different from the way it would be employed in RIP-based
		approaches to 
		compressed sensing.
		To wit, we apply the concept to the
		\emph{expected sensing matrix} (and its inverse), but not to 
		any actual instances.
	}	
	of $X$ is
	\begin{equation*}
		\mathrm{cond}(s,X):=
		\frac{
			\lambda_{\mathrm{max}}(s,X)
		}{
			\lambda_{\mathrm{min}}(s,X)
		}.
	\end{equation*}
\end{definition}

Based on this notion, one can state a strictly stronger version of the
Main Theorem (which is the form we will prove in
Section~\ref{sec:proof}):

\begin{theorem}	\label{kappa_theorem}
	With
	\begin{equation*}
		\kappa_s := \max \left\{ \mathrm{cond}(s,\Sigma),\mathrm{cond}(s,\Sigma^{-1})\right\},
	\end{equation*}
	the conclusion of the main Theorem~\ref{thm:main} continues to hold if
	the lower bound on $m$ is weakened to
	\begin{equation*}
		m \geq C\mu\,\kappa_s\,\omega^2 s \log n,
	\end{equation*}
	for the same constant $C$.
\end{theorem}

We further suspect that the second incoherence condition in
(\ref{eqn:incoherence2}) can be relaxed. Two alternative
bounds not relying on that condition are stated in
Proposition~\ref{prop:incoherence2less} below. (The modifications of
our proof necessary to arrive at these improved estimates will be
sketched after Lemma~\ref{lem:e1}).

\begin{proposition}
	\label{prop:incoherence2less}
	Let $K$ be a constant such that
	\begin{equation*}
		2 \left\| [a a^*, \EE[a a^*]^{-1} ] \right\|_{\infty} \leq K
	\end{equation*}
	holds almost surely, where $[\cdot, \cdot]$ denotes the commutator 
	($[A,B] = AB - BA$)
	and $\| \cdot \|_\infty$ is the operator norm.

	If the requirement (\ref{eqn:incoherence2}) is not necessarily
	fulfilled, the conclusions of Theorem~2 remain valid if the sampling
	rate is bounded below by either
	\begin{equation}
		\label{eqn:2less1}
		m \geq C  \kappa \mu \omega^2 s^2 \log n
	\end{equation}
	or
	\begin{equation}
		\label{eqn:2less2}
		m \geq C ( \kappa \mu s+ K) \omega^2  \log n.
	\end{equation}
\end{proposition}

The commutator bound (\ref{eqn:2less2}) is particularly relevant
for ensembles corresponding to non-uniform samples from an orthogonal
basis. In that case, $\EE[a a^*]$ and $a a^*$ commute with
probability one, so that $K$ may be chosen to be zero. 

There is another degree of freedom which we have not yet systematically
explored: Note that the minimizer of the convex optimization
(\ref{eqn:l1}) does not change if we re-scale \emph{individual}
vectors $a_i \mapsto \nu_i a_i$ for some set of non-zero numbers $\nu_i$. While
we have chosen a \emph{global} scale for the covariance matrix (c.f.\
Lemma~\ref{lem:e0}), the individual weights remain free parameters
that may be used to optimize the sampling rate.  Pursuing this problem
further seems likely to be fruitful.

We remark that the incoherence conditions can be relaxed to hold only
with high probability. This opens up our results to, for example, the
case of Gaussian measurement vectors. The details can be developed in
complete analogy to Ref.~\cite[Appendix B]{candes_probabilistic_2011}.

Lastly, all statements remain true if the measurement vectors are
drawn ``without replacement'' instead of independently --
c.f.~\cite{gross_note_2010} for details.

\section{Relation with previous work and history}
\label{sec:relations}

Most results on sparse vector recovery have relied on certain
conditions that quantify how much a given sampling matrix $A$ distorts
the geometry of the set of all sparse vectors. By far the most
prominent example in that regard is the \emph{restricted isometry
property} (RIP)
\cite{donoho_compressed_2006,rudelson_reconstruction_2011} which
measures the extent to which $A$ deviates from preserving Euclidean distances
between sparse vectors. Conceputally close variations of the RIP
include the \emph{restricted eigenvalue condition} introduced in
\cite{bickel_simultaneous_2009}, or the \emph{restricted correlation
assumption} \cite{bickel_discussion_2007}. Another example is the
\emph{width property} advanced in \cite{kashin_remark_2007}---a Banach
space-theoretic condition that seems to be weaker than the RIP.

From roughly 2008 on, the conceptually strongly related problem of
recovering a low-rank matrix from few expansion coefficients with
respect to a fixed matrix basis has come more and more into focus
\cite{candes_exact_2009,candes_power_2010}. There seems to be no easy
way to directly translate the geometric approaches mentioned above to
the general low-rank matrix recovery problem.  Instead, the pioneering
publications on the matrix problem used fairly elaborate methods from
convex duality theory \cite{candes_exact_2009,candes_power_2010}.
(However, c.f.\
\cite{liu_universal_2011, candes_tight_2011, flammia_quantum_2012} for interesting special cases
where RIP-based techniques \emph{are} applicable to low-rank matrix recovery
problems; and \cite{negahban_restricted_2010} for a related
``restricted strong convexity'' property with consequences for matrix
recovery).

In \cite{gross_quantum_2010,gross_recovering_2011} the second author
and his collaborators introduced a simplified approach to the low-rank
matrix recovery problem. While these works still build on the convex
framework of \cite{candes_exact_2009,candes_power_2010}, they
incorporate several new ideas.
These include the use of non-commutative
large deviation theorems originating from quantum information theory
\cite{ahlswede_strong_2002,tropp_user-friendly_2011}, randomized
constructions based on i.i.d.\ samples of the measurement vectors, and
a certain iterative ``golfing scheme'' for the construction of inexact
dual certificates. These techniques were later modified and adapted to
the original sparse-vector setting in
\cite{candes_probabilistic_2011}. This showed that the conceptual
closeness of the matrix and the vector theory may be used to devise
very similar proofs.

	This ``RIPless'' approach to compressed sensing leads arguably to
	simpler proofs and gives tighter bounds at least for the noise-free
	recovery problem. As far as we know, RIP-based arguments still perform
	superior in the important noisy regime.

	The work \cite{candes_probabilistic_2011} did not include a systematic
	study of non-isotropic ensembles (however, ``small'' deviations from
	isotropy were discussed in Appendix~B). In fact, E.\ Cand\`es
	\cite{candes_probabilistic_2011} suggested to us the problem of
	finding a generalization of the golfing scheme that could cope with
	anisotropic ensembles. This has been achieved by the first author of
	this paper during a research project under the supervision of the
	second author \cite{kueng_richard_efficient_????}. This explains the close relation between
	\cite{candes_probabilistic_2011} and the present work.

	An analysis of anisotropic compressed sensing within the original RIP
	framework has been carried out by other authors, most notably in
	\cite{rudelson_reconstruction_2011}. Since their paper does not
	directly address the noise-free case, a direct comparison of
	statements is difficult. The closest result to ours seems to appear in
	Section~1.3, where a bound of 
	\begin{equation*}
		m \geq \mathcal{O} \big(s M^2 \log n\,\log^3(s \log n)\big)
	\end{equation*}
	for the sampling rate is given. The quantity $M$ is an upper bound on
	the largest coefficient for the measurement vectors $a_i$, related to
	our parameter $\mu$. The
	\emph{big-Oh} notation hides a constant proportional to $\kappa$
	($\rho^{-1}$ in the language of \cite{rudelson_reconstruction_2011}).
	Thus, the basic structure of the solutions is very similar. However,
	some important differences are these:
	\begin{itemize}
		\item
		We do not incur the $\log^3$-term, which is a major advantage of our method.
		Up to a constant factor, our required sampling rate corresponds to the theoretical lower limit.

		\item
		The result in \cite{rudelson_reconstruction_2011} holds
		\emph{uniformly} in the sense that with their probability of
		success, one obtains a sampling matrix which works simultaneously for
		all sparse vectors. This is not the case for us.

		\item
		We have proved no results on noise-resilience. While, following
		\cite{candes_probabilistic_2011}, it should be straight-forward to
		do so, the results may be worse than the RIP-based ones in
		\cite{rudelson_reconstruction_2011}.

		\item
		The proof methods are completely different.
	\end{itemize}
		
	\section{Proof}
	\label{sec:proof}

	The proof is conceptually close to \cite{candes_probabilistic_2011},
	which in turn closely resembles \cite{gross_recovering_2011}. Here we give a largely self-contained presentation.

	\subsection{Notation}
		
		Throughout this paper, we will use the following conventions:	\newline
		If a statements holds almost surely, we will abbreviate this by a.s. 
		In the case of vectors, $\|\cdot\|_p$ denotes the $\ell_p$-norm, whereas
		in the operator case $\|\cdot\|_p$ refers to the Schatten-$p$ norm
		(i.e. the $\ell_p$-norm of the singular values). 
		The letter $z$ will always denote a vector in $\CC^n$ 
	  supported 
		on a set $T$ of cardinality at most $s$ (i.e.\ $z$ is $s$-sparse).
		$T^c$ shall denote the complement of $T$, and $P_T$ ($P_{T^c}$) refers to
		the orthogonal projector onto the set of all vectors supported on
		$T$ ($T^c$). 
		Finally we will use the following technical definitions:
		\begin{equation*}
			X =(\EE[a a^*])^{-1}=\Sigma^{-2}, \qquad
			X_T = P_T X P_T.
		\end{equation*}

	\subsection{Large deviation bounds}

	A central role in the argument is played by certain large deviation
	bounds for sums of matrix-valued random variables. These have been
	introduced in \cite{ahlswede_strong_2002} in the context of quantum information
	theory. The first application to matrix completion and compressed
	sensing problems, as well as the first ``Bernstein version'' taking
	variance information into account, appeared in
	\cite{gross_quantum_2010,gross_recovering_2011}. The
	version we will be making use of derives from Theorem 1.6 in
	\cite{tropp_user-friendly_2011}.

	\begin{proposition}[Matrix Bernstein inequality
	\cite{tropp_user-friendly_2011}]\label{prop:bernstein}
		Consider a finite sequence $\left\{M_k\right\}\in \mathbb{C}^{d\times d}$ of independent, 
		random matrices. Assume that each random matrix satisfies
		$\EE\left[M_k\right]=0$ and $\left\Vert M_k\right\Vert _{\infty}\leq B$ a.s.
		and define 
		\begin{equation*}
		\sigma^2:=
		\max\left\{\|\sum_k\EE\left(M_kM_k^*\right)\|_{\infty},
			\|\sum_k\EE\left(M_k^*M_k\right)\|_{\infty}\right\}.
		\end{equation*}
		Then for all $t\geq0$,
		\begin{equation}\label{eqn:bernstein}
		\operatorname{Pr}\left(\| \sum_{k}M_k\|_{\infty} \geq t\right)
			\leq 2d\exp\left( -\frac{t^{2}/2}{\sigma^{2}+Bt/3}\right) .
		\end{equation}
	\end{proposition}

	We will also require a vector-valued deviation estimate. While one
	could in principle obtain such a statement by applying
	Proposition~\ref{prop:bernstein} to diagonal matrices, a direct
	argument does away with the dimension factor $d$ on the r.h.s.\ of 
	(\ref{eqn:bernstein}). This will save a logarithmic factor in the
	sampling rate of the Main Theorem. The particular vector-valued
	Bernstein inequality below is based on the exposition in 
	\cite{ledoux_probability_1991} (Chapter 6.3, equation (6.12)), with a direct proof appearing in
	\cite{gross_recovering_2011}.

	\begin{proposition}[Vector Bernstein inequality]
		Let 
		$\left\{ g_k\right\} \in\mathbb{C}^{d}$ 
		be a finite sequence of independent random vectors. 
		Suppose that $\EE\left[g_k\right]=0$
		and 
		$\left\Vert g_k\right\Vert _{2}\leq B$ a.s.\  
		and put 
		$\sigma^{2}\geq\sum_{k}\mathbb{E}\left[\left\Vert g_k\right\Vert _{2}^{2}\right]$.
		Then for all $0\leq t\leq\sigma^{2}/B$:
		\begin{eqnarray*}
			\operatorname{Pr}\left(\left\Vert \sum_{k}g_k\right\Vert_2 \geq t\right)  
			&\leq&\exp\left( -\frac{t^{2}}{8\sigma^{2}}+\frac{1}{4}\right) .
		\end{eqnarray*}
	\end{proposition}

	\subsection{Fundamental estimates} \label{sub:fundamental_estimates}

	We adopt the structure and nomenclature of this section from \cite{candes_probabilistic_2011}.
	The following elementary bounds will be used repeatedly:
	\begin{align}
		|\langle a_k, z\rangle|^2    &\leq s \mu \|z\|_2^2,& 
		|\langle a_k, X z \rangle|^2 &\leq s \mu \|z\|_2^2, \\
		\|P_T a_k\|_2^2 &\leq \mu s,&	  
		\|P_T X a_k\|_2^2 &\leq \mu s.
	\end{align}
	Also, we will always assume that $m\geq s$.

	\begin{lemma}[Scaling]
		\label{lem:e0}
		Let $\tilde{a}$ be a random vector such that $\EE[\tilde{a}
		\tilde{a}^*]$ is invertible.

		There is a number $\nu$ such that, with $ a:=\nu \tilde{a}$, it
		holds that
		\begin{equation*}
			\kappa_s =
			\lambda_\mathrm{max}(s,\EE[a a^*]) =
			\lambda_\mathrm{min}(s,\EE[ a a^*])^{-1}
		\end{equation*}
		for all $1\leq s \leq n$.
		This resealed ensemble fulfills:
		\begin{equation}
			\kappa_s \mu	\geq	1.	
		\end{equation}
	\end{lemma}

	\begin{proof}
		The first assertion follows immediately for
		\begin{equation*}
			\nu = \big( 
				\lambda_\mathrm{max}(s, \EE[\tilde{a} \tilde{a}^*]) 
				\lambda_\mathrm{min}(s, \EE[\tilde{a} \tilde{a}^*])
			\big)^{-\frac14}.
		\end{equation*}
		For the second claim: By definition 
		$\mu \geq \max_i |\langle a,e_i \rangle|^2$
		holds almost surely, so that in particular
		\begin{eqnarray*}
			\mu 
			&\geq& \EE \left[ \max_{i} |\langle a,e_i \rangle|^2 \right].
		\end{eqnarray*}
		For every $i$, the function
		\begin{equation*}
			a \mapsto |\langle a, e_i\rangle|^2
		\end{equation*}
		is convex, which implies that
		\begin{equation*}
			a \mapsto\max_{i} |\langle a,e_i \rangle|^2
			= \max_i e_i^* (a a^*) e_i
		\end{equation*}
		is convex (as the pointwise maximum of convex functions). Hence, by
		Jensen's inequality, 
		\begin{eqnarray*}
			\EE \left[ \max_{i} |\langle a,e_i \rangle|^2 \right]
			&\geq& \max_{i} e_i^* \EE[aa^*] e_i	
				= \max_i \langle e_i, \EE[aa^*] e_i \rangle	\\
				&\geq&	\lambda_\mathrm{min}\left(1,\EE[aa^*]\right)	
				\geq	\lambda_\mathrm{min}\left(s,\EE[aa^*]\right).
		\end{eqnarray*}
		Therefore $\mu \geq \lambda_\mathrm{min}\left(s,\EE[aa^*]\right)$.
		Together with $\kappa_s =
		\lambda_\mathrm{min}^{-1}\left(s,\EE[aa^*]\right)$, this implies $\mu
		\kappa_s \geq 1$.	    	
	\end{proof}

	The estimates in this proof are tight in the sense that there are ensembles
	for which each inequality above turns into an equality. A
	straightforward example for such an ensemble is given by picking
	super-normalized Fourier basis vectors $f_k$ (with coefficients $(f_k)_l = e^{2\pi \mathrm{i} \frac{k l}{n}}$)
	according to the uniform probability distribution.

	\begin{lemma}[Local isometry] \label{lem:e1}
		Let $T$ and $P_T$ be as in the notation section. Then for each $0\leq \tau\leq \frac12$:
		\begin{equation*}
			\operatorname{Pr} \left( \|P_T\left( XA^*A-\Id\right)P_T\|_{\infty} \geq \tau \right) \\
			\leq 2s \exp \left(-\frac{m}{s\mu \kappa_s} \frac{\tau^2}{2\left(1+2\tau/3\right)}\right)
		\end{equation*}
	\end{lemma}

	\begin{proof}
		Let us decompose the relevant expression:
		\begin{equation*}
			P_T\left(XA^*A-\Id\right)P_T = \frac{1}{m}\sum_{i=1}^m M_k,
		\end{equation*}
		where $M_k := P_T\left(Xa_ka_k^*-\Id\right)P_T$. Note that $\EE[M_k]=0$. \newline
		We aim to apply the Matrix Bernstein inequality. To this end, we
		estimate
		\begin{eqnarray*}
			\|M_k\|_{\infty} 
			&\leq& \|P_TXa_ka_k^*P_T\|_{2}+1  \\
			&=& \|P_TXa_k\|_2 \, \|a_k^*P_T\|_2+1  \\
			&\leq& \mu s +1 \leq 2 \mu s \kappa_s =:B.
		\end{eqnarray*}	
		Furthermore:
	\begin{eqnarray*}
			 & & \left\|\mathbb{E}\left[M_{k}M_{k}^{*}\right]\right\|_{\infty}  	\\
			&=& 
			\left\Vert \mathbb{E}\left[\left(P_{T}\left(Xa_{k}a_{k}^{*}-\Id\right)
			P_{T}\right)\left(P_{T}\left(a_{k}a_{k}^{*}X-\Id\right)P_{T}\right)\right]\right\Vert_{\infty}\\
			&=& 
			\Big\| \mathbb{E}\left[P_{T}Xa_{k}a_{k}^{*}P_{T}a_{k}a_{k}^{*}XP_{T}\right]
			-\mathbb{E}\left[P_{T}Xa_{k}a_{k}^{*}P_{T}\right]  	
			-\mathbb{E}\left[P_{T}a_{k}a_{k}^{*}XP_{T}\right]+P_{T}\Big\|
			_{\infty} \\
			&=& 
			\left\|\mathbb{E}\left[P_{T}\left(Xa_{k}\left\langle a_{k},P_{T}a_{k}\right\rangle a_{k}^{*}X
			-\Id\right)P_{T}\right]\right\Vert_{\infty} \\
			&\leq& 
			\max\left(\left\Vert \mu s\,\mathbb{E}\left[P_{T}Xa_{k}a_{k}^{*}XP_{T}\right]\right\Vert_{\infty},1\right)\\
			&\leq& \max\left(\mu s\left\Vert X_{T}\right\Vert_{\infty},1\right)
			\leq \max\left(\mu s\kappa_s,1\right) = \mu s
			\kappa_s.
		\end{eqnarray*}

Similarly,
		\begin{eqnarray}
			\left\| \mathbb{E}\left[M_{k}^{*}M_{k}\right]\right\|_{\infty}
			\nonumber 
			&=& 
			\left\| 
				\mathbb{E}\left[P_{T}\left(a_{k}\left\langle a_{k},XP_{T}Xa_{k}\right\rangle a_{k}^{*}
				-I\right)P_{T}\right]
			\right\|_{\infty}\nonumber \\
			 &\leq& 
			\max\left(
				\left\| s\mu\mathbb{E}\left[P_{T}a_{k}a_{k}^{*}P_{T}\right]\right\|_{\infty}
				,1
			\right) \label{eqn:incoherence2used}\\
			 &\leq& 
			 \max\left(
				s\mu\left\| P_{T}X^{-1}P_{T}\right\Vert _{\infty},1
			\right)
			 \leq 
				\mu s\kappa_s \nonumber.
		\end{eqnarray}
		Thus:
		\begin{equation*}
			\max\left\{\|\sum_{k=1}^m\EE\left(M_kM_k^*\right)\|_{\infty},
			\|\sum_{k=1}^m\EE\left(M_k^*M_k\right)\|_{\infty}\right\} 
			\leq ms\mu\kappa_s
			=: \sigma^2.
		\end{equation*}
		Applying the Matrix Bernstein inequality for $s$-dimensional matrices 
		($P_T\left(XA^*A-\Id\right)P_T$ has rank at most $s$)
		with $t=m\tau$ yields the desired result.
	\end{proof}

	The estimate (\ref{eqn:incoherence2used}) is the only place
	in the proof where the second incoherence property in
	(\ref{eqn:incoherence2}) is essentially used. A careful analysis shows
	that in all other cases, one can do without it, possibly at the price
	of replacing $\kappa_s$ by $\kappa$ (which is the reason why we have
	not spelled it out). In order to obtain the results of
	Proposition~\ref{prop:incoherence2less}, the bound
	(\ref{eqn:incoherence2used}) has to be modified. To arrive at
(\ref{eqn:2less1}), use
\begin{eqnarray*}
		\left\| \mathbb{E}\left[M_{k}^{*}M_{k}\right]\right\|_{\infty} 
		&\leq&
 		\mathbb{E} \left[ \left\|\left[M_{k}^{*}M_{k}\right]\right\|_{\infty} \right]\\
		&\leq&
		\EE \left[ \| P_T a_k a_k^* P_T \, \langle a_k, X P_T X a_k\rangle \|_\infty \right] \\
		&\leq&
		s \mu\,\EE  \left[ \langle a_k, X P_T X a_k\rangle \right]
		=
		s \mu\,\EE \left[ \tr \left(a_k a_k^* X P_T X\right) \right]  \\ 
		&=&
		s \mu  \tr \left(X^{-1} X P_T X\right)  
		= 
		s \mu  \tr  \left(P_T X\right)  
		\leq
		s^2 \mu \kappa_s.
\end{eqnarray*}
And for (\ref{eqn:2less2}):
\begin{eqnarray*}
		& & \left\| 
		 	\mathbb{E}\left[P_{T}a_{k}a_{k}^*XP_{T}Xa_{k}a_{k}^{*} P_{T}\right]
		\right\|_{\infty}\nonumber 	\\
		&=&
		\big\| 
		 	\mathbb{E}\left[P_{T}Xa_{k}a_{k}^*P_{T}Xa_{k}a_{k}^{*} P_{T}\right]
		+\mathbb{E}\left[P_{T}[a_{k}a_{k}^*,X]P_{T}Xa_{k}a_{k}^{*} P_{T}\right]
		\big\|_{\infty}\nonumber \\
		&\leq&
		\left\| 
		 	\mathbb{E}\left[P_{T}Xa_{k}a_{k}^*P_{T}a_{k}a_{k}^{*}X P_{T}\right]
		\right\|_{\infty} 
		+2\left\|\mathbb{E}\left[P_{T}[a_{k}a_{k}^*,X]P_{T}Xa_{k}a_{k}^{*} P_{T}\right]
		\right\|_{\infty}\nonumber \\
		&\leq&
		\mu s \kappa_s
		+K \,
		\| \EE\left[ P_T X a_k a_k^* P_T \right] \|_\infty \\
		&=&
		\mu s \kappa_s
		+K \,
		\| P_T X X^{-1} P_T \|_\infty \\
		&=&
		\mu s \kappa_s +K.
\end{eqnarray*}

\begin{lemma}[Low-distortion] \label{lem:e2}
	Let $z, T, P_T$ be as in the notation section.  For each $0\leq \tau\leq 1$ it
	holds that
	\begin{equation*}
		\operatorname{Pr}\big(
			\left\Vert P_{T}\left(\Id-A^{*}AX\right)z\right\Vert _{2}
			\geq \tau \left\Vert z\right\Vert _{2}
		\big)  
		\leq
		\exp\left( -\frac{m\tau^{2}}{16s\mu\kappa_s}+\frac{1}{4}\right).
	\end{equation*}
\end{lemma}

\begin{proof}
	The structure of the proof closely follows the one of
	Lemma \ref{lem:e1}.
	Set
	\begin{equation*}
		g_{k}  :=  P_{T}\left(\Id-a_{k}a_{k}^{*}X\right)z.
	\end{equation*}

	We bound
	\begin{eqnarray*}
		\|g_k\|_2 
		&=& \|P_T(\Id - a_k a_k^* X)z\|_2 \\
		&\leq& \|z\|_2 + \|P_T a_k \langle a_k, X z \rangle \|_2 \\
		&\leq& \|z\|_2 + s \mu \| z\|_2  \leq 2 s \mu \kappa_s \|z\|_2 =: B
	\end{eqnarray*}
	and
	\begin{eqnarray*}
		\EE[\|g_k\|_2^2]
		&\leq& \EE[\|P_T a_k \langle a_k, X z\rangle \|_2^2 ]+\|z\|_2^2 \\
		&=& \EE\big[\|P_T a_k\|_2^2 |\langle a_k, Xz\rangle|^2\big]
			+\|z\|_2^2 \\
		&\leq& s\mu \EE\big[\langle Xz, a_k \rangle\langle a_k, Xz\rangle\big] + \|z\|_2^2 \\
		&=& s\mu \langle X z, \EE[a_k a_k^*] Xz\rangle\big] + \|z\|_2^2\\
		&=& s\mu \langle X z, z\rangle + \|z\|_2^2
		\leq 2s\mu \kappa_s \|z\|_2^2
	\end{eqnarray*}
	so that
	\begin{equation*}
		\sum_{k=1}^m  \EE[\|g_k\|_2^2] 
		\leq  2m s\mu \kappa_s \|z\|_2^2 =: \sigma^2
	\end{equation*}
	and thus $\frac{\sigma^2}{B} = m \|z\|_2$.
	The advertised statement follows by applying the vector Bernstein
	inequality for $t= m\tau$.
\end{proof}

\begin{lemma}[Off-support incoherence] \label{lem:e3}
	Let $z, P_{T^c}$ again be as in the notation section.  Then for each
	$\tau\geq0$: 
	\begin{equation*}
			\operatorname{Pr}\big(
			\left\Vert P_{T^c}A^{*}AX z\right\Vert_\infty
			\geq \tau \left\Vert z\right\Vert _{2}
		\big) 
		\leq 
			2n\exp \left(-\frac{3m\tau^2}{2\mu \kappa_s(3+\sqrt{s}\tau)}\right)
	\end{equation*}
\end{lemma}

\begin{proof}
	Fix $i \in T^{c}$ and use the following decomposition:
	\begin{equation*}
		\langle e_i, A^{*}AXz\rangle
		= \frac{1}{m} \sum_{i=1}^m M_k, 
	\end{equation*}
	where $M_k := \langle e_i, a_k a_{k}^{*} Xz \rangle = \langle e_i,a_k \rangle \langle a_k,Xz \rangle$. 
	Note that we have:
	\begin{equation*}
		\EE [M_k] = \langle e_i ,\EE[a_k a_k^{*}]Xz\rangle = \langle e_i,z\rangle = 0, 
	\end{equation*}
	because $e_i \in T^c$. 
	Bound
	\begin{equation*}
	|M_k| = |\langle e_i,a_k \rangle \langle a_k,Xz\rangle| \leq \sqrt{s}\mu \kappa_s \|z\|_2=:B,
	\end{equation*}
	and
	\begin{eqnarray*}
		\EE[M_k M_k^*] 
		&=& \EE[M_k^* M_k] = \EE[|\langle a_k,e_i\rangle |^2 |\langle a_k,Xz\rangle|^2] \\
		&\leq& \mu \EE[\langle Xz,a_k a_k^* Xz\rangle] = \mu \langle Xz,z\rangle  \\	
		&\leq& \mu \|X_T\|_{\infty} \|z\|_2^2 \leq \mu \kappa_s\|z\|_2^2. \\
	\end{eqnarray*} 
	Therefore we can set $\sigma^2 := m \mu  \kappa_s \|z\|_2^2$. 
	Applying the Matrix Bernstein inequality for $d=1$ and the union bound
	over all $i\in T^c$ yields the claim.
\end{proof}

\begin{lemma}[Uniform off-support incoherence] \label{lem:e4}
	Let $T^c, P_T$ be as in the notation section.  For $0\leq\tau\leq1$ we have
	\begin{equation*}
		\operatorname{Pr} \left( \max_{i\in T^c} \|P_TXA^*Ae_i\|_2 \geq \tau \right)
		\leq	n\exp \left(-\frac{m\tau^2}{8s\mu\kappa_s}+\frac{1}{4}\right)
	\end{equation*}
\end{lemma}

\begin{proof}
	Fix $i\in T^c$ and decompose:
	\begin{equation*}
		P_TXA^*A e_i = \frac{1}{m}\sum_{k=1}^m g_k,
	\end{equation*}
	where $g_k := \langle a_k, e_i \rangle P_TXa_k$. It holds that
	$\EE[g_k]=0$. Next, bound
	\begin{equation*}
		\|g_k\|_2 = |\langle a_k,e_i \rangle| \|P_TXa_k\|_2
		\leq s\mu=: B.
	\end{equation*}
	Furthermore:
	\begin{equation*}
		\EE[\|g_k\|_2^2]	
		\leq \sum_{i\in T} \mu\EE[\langle e_i,Xa_ka_k^*Xe_i\rangle]	
		\leq \sum_{i\in T} \mu\|X_T\|_{\infty}\leq s\mu\kappa_s.	
	\end{equation*}	
	We can therefore set $\sigma^2:= ms\mu \kappa_s$ and apply
	the Vector Bernstein inequality for $t=m\tau$.  Noting that
	$\sigma^2/B=m\kappa_s\geq m$ finishes the proof.
\end{proof} 

\subsection{Convex geometry}\label{sub:Convex_Geometry}

Our aim is to prove that the solution $x^{\star}$ to the optimization
problem (\ref{eqn:l1}) equals the unknown vector $x$.
One way of assuring this is by exhibiting a \emph{dual
certificate} \cite{bertsekas_convex_2003}.
This
method was first introduced in \cite{candes_near-optimal_2006} and is
now standard. We will use a relaxed version of this first introduced in
\cite{gross_recovering_2011} and later adapted from matrices to
vectors in \cite{candes_probabilistic_2011}. Our version further
adapts the statement to the anisotropic setting.

\begin{lemma}[Inexact duality]\label{lem:inexact}
	Let $x\in\CC^n$ be a $s$-sparse vector, let
	$T=\mathrm{supp}\left(x\right)$.

	Assume that
	\begin{eqnarray}
		\|\left(P_T XA^*AP_T\right)^{-1}\|_{\infty} 
		&\leq& 2, \label{eq:ass1}\\
		\mathrm{max}_{i\in T^c}\|P_T XA^*Ae_i\|_2 
		&\leq& 1 \label{eq:ass2}
	\end{eqnarray}
	and that there is a vector $v$ in the row space of $A$ obeying
	\begin{eqnarray}
		\|v_T-\mathrm{sgn}\left(x\right)\|_2	&\leq& \frac{1}{4} 
		\label{eq:dual1}\\
		\|v_{T^c}\|_{\infty}			&\leq& \frac{1}{4}.
		\label{eq:dual2} 
	\end{eqnarray}

	Then the solution $x^\star$ of the convex program (\ref{eqn:l1}) is
	unique and equal to $x$.
\end{lemma}

\begin{proof}
	Let $\hat{x}=x+h$ be a solution of the minimization procedure.
	We note that feasibility requires $Ah=0$.
	To prove the claim it suffices to show $h=0$. Observe:
	\begin{eqnarray*}
		\|\hat{x}\|_1
		&=&	\|x+h_T\|_1+\|h_{T^c}\|_1 \\
		&=&	\langle \mathrm{sgn}\left(x+h_T\right),x+h_T \rangle 		+	\|h_{T^c}\|_1\\
		&\geq&	\langle \mathrm{sgn}\left(x\right),x \rangle
		+	\langle \mathrm{sgn}\left(x \right),h_T \rangle
		+	\|h_{T^c}\|_1 \\
		&\geq&	\|x\|_1
		-	|\langle \mathrm{sgn}\left(x\right),h_T\rangle|
		+	\|h_{T^c}\|_1.
	\end{eqnarray*}
	Feasibility requires $\langle v,h \rangle = 0$ (since $v$ is in the row space of $A$)
	and therefore:
	\begin{eqnarray*}
		|\langle \mathrm{sgn}\left(x\right),h_T \rangle|
		&=&	|\langle \mathrm{sgn}\left(x\right)-v_T,h_T \rangle
		+	 \langle v_T,h_T \rangle| \\
		&=&	|\langle \mathrm{sgn}\left(x\right)-v_T,h_T\rangle
		-	\langle v_{T^c},h_{T^c} \rangle| \\
		&\leq&	|\langle \mathrm{sgn}\left(x\right)-v_T,h_T\rangle|
		+	|\langle v_{T^c},h_{T^c} \rangle| \\
		&\leq&	\| \mathrm{sgn}\left( x \right)-v_T \|_2
			\|h_T\|_2
			+ |\langle v_{T^c},h_{T^c} \rangle | \\
		&\leq&	\frac{1}{4}\|h_T\|_2
			+ |\langle v_{T^c},h_{T^c} \rangle |, 
	\end{eqnarray*}
	where we have used (\ref{eq:dual1}). Together with:
	\begin{equation*}
		|\langle v_{T^c},h_{T^c} \rangle|
		\leq	\|v_{T^c}\|_{\infty} \|h_{T^c}\|_1 
		\leq	\frac{1}{4}\|h_{T^c}\|_1,
	\end{equation*}
	this implies:
	\begin{equation*}
		|\langle \mathrm{sgn} \left(x\right),h_T \rangle|
		\leq	\frac{1}{4}\left(\|h_T\|_2 + \|h_{T^c}\|_1 \right).
	\end{equation*}	
	Furthermore due to (\ref{eq:ass1}) and (\ref{eq:ass2}) it holds that		
	\begin{eqnarray*}
		\|h_T\|_2
		&=&	\|\left( P_TXA^*AP_T\right)^{-1} \left(P_TXA^*AP_T\right) h_T\|_2 \\
		&=&	\|\left( P_TXA^*AP_T\right)^{-1} \left(P_TXA^*A\right)
			\left(h-h_{T^c}\right) \|_2\\
		&=&	\|-\left( P_TXA^*AP_T\right)^{-1} \left(P_TXA^*A\right) h_{T^c}\|_2 \\
		&\leq&	2\|P_TXA^*AP_{T^c}h\|_2 \\ 
		&\leq&	2\mathrm{max}_{i\in T^c}
			\|P_TXA^*Ae_i\|_2 \|h_{T^c}\|_1 \\
		&\leq&	2\|h_{T^c}\|_1, 
	\end{eqnarray*}
	All this together implies:
	\begin{eqnarray*}
		\|\hat{x}\|_1
		&\geq&	\|x\|_1 -\frac{1}{4}\|h_T\|_2 + \frac{3}{4}\|h_{T^c}\|_1 \\
		&\geq&	\|x\|_1 + \frac{1}{4}\|h_{T^c}\|_1.
	\end{eqnarray*}
	Consequently $\|\hat{x}\|_1 = \|x\|_1$ demands $\|h_{T^c}\|_1 = 0$,
	which in turn implies $\|h_T\|_2 = 0$, because
	$\|h_T\|_2\leq 2\|h_{T^c}\|_1$.
	Therefore $h=0$ which corresponds to a unique minimizer ($\hat{x}=x$).
	\end{proof}

\subsection{Construction of the certificate} \label{sub:Construction}

It remains to show that a dual certificate $v$ as described in
Lemma~\ref{lem:inexact} can indeed be constructed. We will
prove:

\begin{lemma}
	Let $x\in\CC^n$ be an $s$-sparse vector, let $\omega\geq1$. If the
	number of measurements fulfills
	\begin{equation*}
		m \geq 18044 \kappa_s \mu \omega^2 s \log n,
	\end{equation*}
	then with probability at least 
	$1-\mathrm{e}^{-\omega}$, the constraints (\ref{eq:ass1},
	\ref{eq:ass2}) will hold and a vector $v$ with the properties
	required for Lemma~\ref{lem:inexact} exists.
\end{lemma}

	This lemma immediately implies the Main Theorem. \newline
	The proof employs a recursive procedure (dubbed the ``golfing
	scheme'') to construct a sequence $v_i$ of vectors converging to a
	dual certificate with high probability. The technique has been
	developed in \cite{gross_quantum_2010,gross_recovering_2011} in the
	context of low-rank matrix recovery problems and has later been refined for
	compressed sensing in \cite{candes_probabilistic_2011}. Here, we further
	modify the construction to handle anisotropic ensembles.

	\begin{proof}
	The recursive scheme consists of $l$ iterations. The $i$-th iteration
		depends on three parameters: $m_i\in\NN; c_i, t_i\in\RR$
		which will be chosen in the course of the later analysis.
		To initialize, set
		\begin{equation*}
			v_0 = 0
		\end{equation*}
		(the $v_i$ for $1\leq i \leq l$ will be defined iteratively below).
		We will use the notation
		\begin{equation*}
			q_i = \mathrm{sgn}\left(x\right) - P_Tv_i.
		\end{equation*}

		The $i$-th step of the scheme proceeds according to the following
		protocol: We sample $m_i$ vectors from the ensemble $F$. Let
		$\tilde A$ be the $m_i \times n$-matrix whose rows consists of
		these  vectors. We check whether the following two conditions
		are met:
		\begin{eqnarray}
			\|P_T\left(\Id-\frac{m}{m_i}\tilde A^* \tilde A X\right)P_Tq_{i-1}\|_2
			&\leq&	c_i\|q_{i-1}\|_2, \label{eq:golfing1}\\
			\|\frac{m}{m_i}P_{T^c}\tilde A^* \tilde A XP_Tq_{i-1}\|_\infty
			&\leq&	t_i\|q_{i-1}\|_2. \label{eq:golfing2}
		\end{eqnarray}	
		If so, set 
		\begin{equation*}	
			A_i = \tilde A,	\qquad 
			v_i =	\frac{m}{m_i}A^*_iA_iXP_T\left(
			\mathrm{sgn}\left(x\right)-v_{i-1}\right) + v_{i-1}
		\end{equation*}
		and proceed to step $i+1$.
		If either of (\ref{eq:golfing1}), (\ref{eq:golfing2}) fails to
		hold, repeat the $i$-th step with a fresh batch of $m_i$ vectors
		drawn from $F$. Denote the number of repetitions of the $i$-th step
		by $r_i$.

		We now analyze the properties of the above recursive
		construction.
		The following identities are easily verified by repeating the given transformations inductively:
		\begin{eqnarray}
			v&:=&v_l
			= \frac{m}{m_l}A_l^*A_l X P_T \left( \mathrm{sgn}(x)-v_{l-1}\right) + v_{l-1} 	\nonumber \\
			&=& \frac{m}{m_l}A_l^*A_l X P_T q_{l-1} + v_{l-1}	\nonumber \\
			&=& \ldots = \sum_{i=1}^l \frac{m}{m_i}A^*_iA_iXP_Tq_{i-1}, \label{eq:v}\\
			q_i &=&	\mathrm{sgn}(x) - P_T v_i 	\nonumber \\ 
			&=& \mathrm{sgn}(x) - P_T \left( \frac{m}{m_i}A_i^* A_i X P_T \left( \mathrm{sgn}(x) - v_{i-1} \right) + v_{i+1} \right) \nonumber \\
			&=& \left(\mathrm{sgn}(x) - P_T v_{i-1}\right) - \frac{m}{m_i} A_i^*A_i X P_T \left( \mathrm{sgn}(x) - v_{i-1} \right) \nonumber \\
			&=& P_T \left( \Id - \frac{m}{m_i}A_i^* A_i X \right) q_{i-1}  \nonumber \\
			&=&  \ldots = \prod_{j=1}^iP_T\left(\Id - \frac{m}{m_i}A^*_jA_jX\right) P_T\mathrm{sgn}\left(x\right). \label{eq:q}
		\end{eqnarray}
		Together with
		(\ref{eq:golfing1}) and (\ref{eq:golfing2}), one obtains
		\begin{eqnarray*}
			\|q_l\|_2
			&\leq&	c_l\|q_{l-1}\|_2
			\leq	\prod_{i=1}^lc_i\|q_0\|_2
			=	\prod_{i-1}^lc_i\|\mathrm{sgn}\left(x\right)\|_2 
			=	\sqrt{s}\prod_{i=1}^lc_i, \\
			\|v_{T^c}\|_{\infty} 
			&=&	\left\|P_{T^c}\left(\sum_{i=1}^l
				\frac{m}{m_i}A^*_iA_iXP_Tq_{i-1}\right) \right\|_{\infty}  \\
			&\leq&	\sum_{i=1}^l \left\|\frac{m}{m_i}P_{T^c}A^*_iA_iXP_Tq_{i-1}\right\|_2\\
			&\leq&	\sum_{i=1}^lt_i\|q_{i-1}\|_2
			\leq	\sqrt{s}\left(t_1+\sum_{i=2}^lt_i\prod_{j=1}^{i-1}c_j\right).
		\end{eqnarray*}
		Following \cite{gross_recovering_2011}, we choose the parameters
		$l, c_i, t_i$ as
		\begin{equation*}
			l = \left\lceil \frac{1}{2} \log_2 s \right\rceil + 2,	\qquad
			c_1 = c_2 =
				\frac{1}{2\sqrt{\log n}},  \qquad
			t_1 = t_2 = 
				\frac{1}{8\sqrt{s}},
		\end{equation*}
		and for $i\geq 3$
		\begin{eqnarray*}
				t_i = \frac{\log n}{8\sqrt{s}}, \qquad
				c_i = \frac{1}{2}.
		\end{eqnarray*}
		A short calculation then yields
		\begin{equation*}
			\|v_{T^c}\|_\infty \leq \frac14,
			\qquad
			\|v - \mathrm{sgn}(x_T)\|_2 = \|q_l\|_2 \leq \frac14,
		\end{equation*}
		which are conditions (\ref{eq:dual1}) and (\ref{eq:dual2}).

		Next, we need to establish that the total number
		\begin{equation*}
			\sum_{i=1}^l m_i r_i
		\end{equation*}
		of sampled vectors remains small with high probability. More
		precisely, we will bound the probability
		\begin{equation*}
			p_3 :=
			\operatorname{Pr}
			\left(
				(r_1 > 1)
				\text{ or }
				(r_2 > 1)
				\text{ or }
				\sum_{i=1}^l r_i \geq l' 
			\right)
		\end{equation*}
		for some $l'$ to be chosen later.

		To that end, 	
		denote by $p_1(i)$ the probability that (\ref{eq:golfing1})
		fails to hold in any given batch of the $i$-th step. Analogously,
		let $p_2(i)$ be the probability of failure for
		(\ref{eq:golfing2}). Lemmas \ref{lem:e2} and \ref{lem:e3} give the estimates
		\begin{equation*}
			p_1 \left(i\right)
			\leq	\mathrm{exp} \left(-\frac{m_ic_i^2}{16s\mu \kappa_s}
			+\frac{1}{4}\right), \qquad
			p_2\left(i\right)
			\leq	2n\exp \left( -\frac{3m_i t_i^2}
				{2\mu \kappa_s\left(3+\sqrt{s}t_i\right)} \right).
		\end{equation*}
		We choose
		\begin{equation*}
			l' = 4(\omega+\log 12 + \frac23 l), \qquad 
			m_1=m_2 =
				694  \kappa_s \mu \omega s \log n,
		\end{equation*}
		and for $i\geq 3$
		\begin{equation*}
			m_i	=
			694 \kappa_s \mu \omega s.
		\end{equation*}
		Such a choice can be guaranteed by a total sampling rate
		$m \geq 18044  \kappa_s \mu \omega^2 s \log n$
		and ensures
		\begin{equation*}
			p_1(i) + p_2(i) \leq \frac16 e^{-\omega} \leq \frac{1}{12}
		\end{equation*}
		for all $i$. 
		(It is easily seen that for for $n\gg1$, a bound of
		$m\geq 228 \kappa_s \mu \omega^2 s \log n$ is sufficient. The constants
		appearing here are highly unlikely to be optimal.)
		Note that 
		\begin{equation*}
				\sum_{i=1}^l r_i \geq l' 
		\end{equation*}
		only if fewer than $l$ of the first $l'$ batches of vectors
		satisfied both (\ref{eq:golfing1}) and (\ref{eq:golfing2}). This
		implies that
		\begin{equation*}
			\operatorname{Pr}\left(
				\sum_{i=1}^l r_i \geq l' 
			\right)
			\leq
			\operatorname{Pr}(N\leq l-1)_{\operatorname{Bin}(l',\frac{11}{12})},
		\end{equation*}
		where the r.h.s.\ is the probability of obtaining fewer than $l$
		outcomes in a binomial process with $l'$ repetitions and individual
		success probability $11/12$. We bound this quantity using a
		standard concentration bound from
		\cite{mcdiarmid_colin_concentration_????} (C. McDiarmid's section "Concentration"):
		\begin{equation*}
			\operatorname{Pr}\left(|\mathrm{Bin}\left(n,p\right)
			-np|>\tau\right)
			\leq
			2\,\mathrm{exp}\left(-\frac{\tau^2}{3np}\right). 
		\end{equation*}
		 This yields $ \operatorname{Pr}\left(
				\sum_{i=1}^l r_i \geq l' 
			\right)\leq\frac16 \mathrm{e}^{-\omega}$ for our choice of
		$l^{'}$.
		Putting things together, we have
		\begin{equation*}
			p_3 \leq 3\,\frac16 \mathrm{e}^{-\omega} 
			= \frac12 \mathrm{e}^{-\omega}
		\end{equation*}
		according to the union bound.
		In addition, we have to take into account that properties
		(\ref{eq:ass1}) and (\ref{eq:ass2}) can fail as well.
		We denote these probabilities of failure by $p_4$ and $p_5$.
		Lemmas \ref{lem:e1} and \ref{lem:e4} give:
		\begin{equation*}
		p_4
		\leq 2s \mathrm{exp}\left(-\frac{6m}{7s\mu \kappa_s}\right), \qquad
		p_5
		\leq	n \mathrm{exp}\left(
		-\frac{m}{8s\mu \kappa_s}+\frac14 \right).
		\end{equation*}
		Our sampling rate $m$ guarantees $p_4 \leq \frac14\mathrm{e}^{-\omega}$ as well as
		$p_5 \leq \frac14\mathrm{e}^{-\omega}$.
		Applying the union bound now yields our desired
		 overall error bound 
			($p_3+p_4+p_5\leq \mathrm{e}^{-\omega}$).
\end{proof}

\section{Conclusion and Outlook}

In this paper, we have shown that proof techniques based on duality
theory and the ``golfing scheme'' are versatile enough to handle the
situation where the ensemble of measurement vectors is not isotropic.

An obvious future line of research would be to translate these results
to the low-rank matrix recovery problem. Given the high degree of
similarity between \cite{gross_recovering_2011} and
\cite{candes_probabilistic_2011}, this should be a conceptually
straight-forward task. This would further generalize the scope of this
proof method, beyond ortho-normal operator bases
\cite{gross_recovering_2011} and tight frames
\cite{ohliger_continuous-variable_2011}.

Also, Proposition~\ref{prop:incoherence2less} suggests that the second
incoherence property (\ref{eqn:incoherence2}) can be relaxed or maybe
even disposed of. We leave this as an open problem.

\section{Acknowledgments}

We thank E.~Cand\`es for suggesting the problem treated here, and him, 
Y.~Plan, P.~Jung, and P.~Walk for insightful discussions.  Financial
support from the Excellence Initiative of the German Federal and State
Governments (grant ZUK 43), the German Science Foundation (DFG grants
CH 843/1-1 and CH 843/2-1), the Swiss National Science Foundation, and
the Swiss National Center of Competence in Research ``Quantum Science and
Technology'' is gratefully acknowledged.

\bibliographystyle{elsarticle-num}
\bibliography{anisotropic1}

\end{document}